\documentclass[11pt]{article}
\usepackage{setspace} 
\usepackage{amsmath}
\usepackage{amssymb}
\usepackage{cite}
\usepackage{rotating,makecell,color,subfig,verbatim,xspace}
\usepackage{pdfsync}
\usepackage{latexsym, amsthm, amssymb,amsmath}
\makeatletter
\renewcommand{\@biblabel}[1]{\quad#1.}
\makeatother

\date{}
\newcommand{\var}[1]{\mathrm{Var}\left( #1 \right)}

\newtheorem{observation}{Observation}

\begin{document}
\begin{flushleft}
{\Large\textbf
A general framework
for meta-analyzing dependent studies with overlapping subjects
in association mapping
}\\
Buhm Han\,$^{1,2,3,4,\ast}$, Jae Hoon Sul\,$^{5}$, 
Eleazar Eskin\,$^{5,6}$,
Paul I. W. de Bakker\,$^{1,4,7,8}$, 
Soumya Raychaudhuri\,$^{1,2,3,4,9}$\\

$^1$Division of Genetics, Brigham and Women's Hospital, Harvard Medical School, Boston, Massachusetts, USA\\
$^2$Division of Rheumatology, Brigham and Women's Hospital, Harvard Medical School, Boston, Massachusetts, USA\\
$^3$Partners Center for Personalized Genetic Medicine, Boston, Massachusetts, USA\\
$^4$Program in Medical and Population Genetics, Broad Institute of Harvard and MIT, Cambridge, Massachusetts, USA\\
$^5$Computer Science Department, University of California, Los Angeles, California, USA\\
$^6$Department of Human Genetics, University of California, Los Angeles, California, USA\\
$^7$Julius Center for Health Sciences and Primary Care, University Medical Center Utrecht, Utrecht, The Netherlands\\
$^8$Department of Medical Genetics, University Medical Center Utrecht, Utrecht, The Netherlands\\
$^9$Faculty of Medical and Human Sciences, University of Manchester, Manchester, UK\\

 $^\ast$ Corresponding Author E-mail: buhmhan@broadinstitute.org\\
\end{flushleft}

\begin{abstract}
Meta-analysis of genome-wide association studies 
is increasingly popular and many meta-analytic methods have been recently proposed.
A majority of meta-analytic methods combine information from multiple studies by assuming that studies are independent since individuals collected in one study are unlikely to be collected again by another study. 
However, it has become increasingly common to utilize the same control individuals among multiple studies to reduce genotyping or sequencing cost. This causes those studies that share the same individuals to be 
dependent, 
and spurious associations may arise if overlapping subjects are not taken into account in a meta-analysis. 
In this paper, we propose a general framework for meta-analyzing dependent studies 
with overlapping subjects.
Given dependent studies, our approach ``decouples'' the studies into independent studies
such that meta-analysis methods assuming independent studies can be applied.
This enables many meta-analysis methods, such as the random effects model,
to account for overlapping subjects.
Another advantage is that one can continue to use preferred software in the analysis pipeline
which may not support overlapping subjects.
Using simulations and the Wellcome Trust Case Control Consortium data,
we show that our decoupling approach allows both the fixed and the random effects models
to account for overlapping subjects while retaining desirable false positive rate and power.
\end{abstract}

\clearpage
\section{Introduction}

Meta-analysis of genome-wide association studies is becoming increasingly popular
\cite{Fleiss93,Bakker08}.
Investigators combine multiple studies into a single meta-analysis
to increase sample size and thereby the statistical power to detect causal variants. 
In the past two to three years, 
large-scale meta-analysis have been highly successful dramatically increasing
the number of known associated loci in many human diseases \cite{Eyre12,Jostins12}.

There exist a number of categories of methods in meta-analysis.
The fixed effects model (FE) assumes that the effect sizes are fixed across the studies
and is powerful when the assumption holds \cite{Cochran54,Mantel59}.
The random effects model (RE) assumes that the effect sizes can be different between studies,
a phenomenon called heterogeneity \cite{DerSimonian86}.
A recently proposed random effects model is shown to be more powerful than FE
if the data are heterogeneous \cite{Han11}.
 Additional categories of methods include
the p-value-based approaches \cite{Fisher67,Zaykin10},
the subset approaches assuming that the effects are present or absent in the studies \cite{Bhattacharjee12,Han12},
and the Bayesian approaches \cite{Wen11,Morris11,Wen12}.
A majority of these methods combine information from multiple studies 
by assuming that studies are independent 
since individuals collected in one study are unlikely to be collected again by another study.

However, it has become increasingly common to utilize the same control individuals among multiple studies to reduce genotyping or sequencing cost \cite{Burton07}.
This causes those studies that share the same individuals to be dependent,
 and spurious associations may arise if overlapping subjects are not taken into account in a meta-analysis.
A naive solution would be to manually split the overlapping subjects
into distinct studies,
which can be sub-optimal \cite{Lin09} and 
may not be practical if genotype data are not shared.

Recently, Lin and Sullivan proposed a meta-analytic approach 
that takes into account overlapping subjects \cite{Lin09}.
This approach provides an optimal and elegant solution to 
account for the correlation structure between studies caused by the overlapping subjects.
However, 
their method is exclusively based on the fixed effects model.
Recent studies extended this method to the p-value-based approach \cite{Zaykin10}
and the subset approach \cite{Bhattacharjee12},
but to date,
it is unclear how to account for overlapping subjects
in the random effects model and other meta-analytic approaches.

In this paper, 
we propose 
a general framework for meta-analyzing dependent studies with overlapping subjects,
which is an extension of the Lin and Sullivan approach \cite{Lin09}.
Given the correlation structure between studies,
the core idea is
to uncorrelate or decouple the studies into independent
studies 
such that meta-analysis methods assuming independent studies can be applied.
The advantage of our decoupling approach is that it enables
many meta-analysis methods, such as the random effects model,
to account for overlapping subjects.
Since our approach involves only the data-side change rather than the method-side change, 
one can continue to use preferred software 
in the existing analysis pipeline 
which may not support overlapping subjects.
We analytically show that the Lin and Sullivan approach is one special case in our framework in which the fixed effects model is applied after decoupling studies.
We demonstrate the utility of our approach by performing 
a meta-analysis of three autoimmune diseases using
the Wellcome Trust Case Control Consortium data \cite{Burton07}.

\section{Results}

\subsection{Overview of the Method}

Many traditional as well as recently proposed meta-analytic methods exist,
but a majority of them cannot account for overlapping subjects
(Table \ref{table:methods}).
Moreover, even if 
there exists a solution for overlapping subjects in a specific method
such as the fixed effects model,
the software one prefers may not support overlapping subjects.
For example, widely used software METAL \cite{Metal:web} and MANTEL \cite{Bakker08} do not support
overlapping subjects.
We propose a general
framework that allows different approaches to 
deal with overlapping subjects, which extends
the Lin and Sullivan approach \cite{Lin09}.
Through this paper, we will often use the term ``correlation of studies''
to refer to the correlation of statistics (typically z-scores) of studies in short.
The intuition is that the more studies are correlated, 
the less information they contain with respect to
the summary statistic of the meta-analysis.
For example, if two studies are perfectly correlated ($r=1$),
their combined information is not better than
a single study's information.
In Figure \ref{fig:schematic}, we have three studies A, B, and C
whose statistics are correlated.
For simplicity, their variances are set to 1.0.
Our approach ``decouples'' the studies into independent studies
that have the same information with respect to the summary statistic.
The penalty for the decoupling is the increased variances.
The variance of the study B has increased the most drastically (2.52), 
because its correlations to A and C were large (0.5 and 0.3).
The size of the circles denotes the amount of information 
in terms of the inverse variance, showing that B has the smallest information.
We can then use the decoupled studies in the downstream meta-analysis method.
If the downstream method is the fixed effects model,
our decoupling approach is equivalent to the optimal method of Lin and Sullivan \cite{Lin09}
(See Methods).

\subsection{False positive rate and power simulations}

We perform simulations to examine the false positive rate and power of our decoupling approach.
We suppose that ten different studies are combined in a meta-analysis
to test a genotyped marker.
We make an assumption that the studies are uniform in their sample sizes and 
the marker allele frequencies.
We also assume that the sample sizes are sufficiently large.
Under these assumptions, 
simulating genotype data is approximately equivalent
to simulating the observed effect sizes directly from a normal distribution.
For simplicity, we assume that the variances of effect sizes are uniformly 1.0.
We use the significance threshold $\alpha = 0.05$ for all simulations below.

We first simulate the null model
in which the marker is not associated with a disease in all studies.
We assume that the correlation $r_{ij}$ between study $i$ and $j$ is uniform for all
study pairs $i \ne j$.
This defines our covariance matrix $\mathbf{\Omega}$.
Then we sample the vector of observed effect sizes $\hat{\mathbf{x}}$ from $N(0, \mathbf{\Omega})$ 10,000 times.
We vary $r_{ij}$ from 0.0 to 0.9 and measure the false positive rates for different meta-analytic 
approaches.

In Figure \ref{fig:power}A,
we compare the false positive rate of the methods for the fixed effects model.
The naive FE method refers to the traditional fixed effects model
unaccounting for the correlations.
The naive method shows dramatically inflated false positive rate as expected,
since the correlations are ignored.
The false positive rate becomes exacerbated as the unaccounted correlation increases,
up to $0.52$ at $r_{ij}=0.9$.
The decoupling FE refers to our decoupling approach 
applying FE after decoupling the studies.
Both the decoupling FE and the Lin-Sullivan approach correctly 
control the false positive rate.
The two methods yield the identical results,
since they are equivalent (See Methods).
The average false positive rate of the two methods over all correlation values $r_{ij}$ was 
identically 0.050.

In Figure \ref{fig:power}B,
we assess the false positive rate of the methods for the random effects model.
The naive RE method refers to the Han and Eskin random effects model \cite{Han11}
unaccounting for the correlation.
The naive method shows dramatically inflated false positive rate that increases with
the correlation.
The decoupling RE refers to our decoupling approach 
applying 
the Han and Eskin random effects model \cite{Han11}
 after decoupling the studies.  
The decoupling RE correctly controls the false positive rate,
with some conservative tendencies.
The average false positive rate of the decoupling RE over all correlation values $r_{ij}$ was 0.034.

In Figure \ref{fig:power}C,
we simulate the alternative model assuming 
that the fixed effects model is the generative model.
We fix the correlation to be $r_{ij}=0.20$. 
We sample $\hat{\mathbf{x}}$ from $N(\beta\mathbf{e}, \mathbf{\Omega})$ 
where we vary the mean effect $\beta$ from 0.0 to 2.0.
The decoupling FE shows power increase as the mean effect increases.
The power is identical to the Lin-Sullivan approach,
since the two methods are equivalent.
Note that the naive FE method is not shown in the power comparison
since its false positive rate is not properly controlled.

In Figure \ref{fig:power}D,
we simulate the alternative model assuming 
that the random effects model is the generative model.
Again, we fix the correlation to be $r_{ij}=0.20$. 
We sample $\hat{\mathbf{x}}$ from $N(\beta\mathbf{e}, \mathbf{\Omega} + \tau^2I)$ 
where we vary both $\beta$ and the heterogeneity $\tau^2$.
The power of the decoupling RE is shown for different configurations of the models.
We find that the power shows typical characteristics of the random effects model;
the power increases as the mean effect increases and as the heterogeneity increases \cite{Han11}.
This shows that when we direct decoupled studies into the random effects model,
the method has power to detect alternative models that the random effects model is designed for.

In summary, our simulations stress a few points;
(1) The decoupling approach can be flexibly applied to both the fixed and the random effects models.
(2) When applied to the fixed effects model, the decoupling approach shows equivalent results
to the Lin-Sullivan approach. The method accurately controls the false positive rate
and shows power to detect the alternative model.
(3) When applied to the random effects model, the decoupling approach controls the false
positive rate with conservative tendencies, while retaining the power to detect alternative models
of the random effects model.

\subsection{Applications to the Wellcome Trust Case Control Consortium data}

We apply our decoupling approach to the Wellcome Trust Case Control Consortium (WTCCC) data. 
The WTCCC has performed genome-wide association studies of seven diseases 
(bipolar disorder, coronary artery disease, Crohn's disease, hypertension, rheumatoid arthritis,
type 1 diabetes, and type 2 diabetes, or BD, CAD, CD, HT, RA, T1D, and T2D in short).
Using these data, we perform meta-analysis of three autoimmune diseases (CD, RA, and T1D).
These data sets are a good example of overlapping subjects because all of the controls are shared
between diseases.
The WTCCC performed a combined analysis using the genotype data of these
three diseases and reported four significant loci 
(See Supplementary Table 11 of Burton et al. \cite{Burton07}).
We want to show the utility of our approach 
by reproducing the same results 
only
using summary statistics without genotype data.

We first calculate the correlation matrix between the seven diseases using the Lin and Sullivan formula
(See equation (\ref{eq:r}) in Methods).
Figure \ref{fig:cor} shows that the studies are positively correlated due to the shared control design.
The correlations are at around $r=0.4$ at all pairs of the diseases. 
These uniform correlations reflects
the unique study design that all controls are shared and the similar numbers of cases are collected in
all diseases.

We perform the meta-analysis of three autoimmune diseases (CD, RA, and T1D) using
the log odds ratios and their standard errors.
We consider 397,450 SNPs that passed quality control for all three diseases and 
the minor allele frequency is greater than 1\%.
We first apply the naive fixed effects model (FE) and the random effects model (RE) 
which do not take into account the correlation structure.
Figure \ref{fig:wtcccqq}A shows that the qq-plot is highly inflated for both the naive FE and 
the naive RE
(the genomic control factors \cite{Devlin95}, $\lambda_{\text{FE}}=1.86$ and $\lambda_{\text{RE}}=1.62$,
excluding the MHC region).
Since the p-values are highly inflated,
further downstream analyses using these naive approaches can be susceptible to false positives.

We then apply our decoupling approach to account for the correlation structure.
We construct the decoupled studies and apply FE and RE.
Figure \ref{fig:wtcccqq}B shows that the qq-plot is much better calibrated
($\lambda_{\text{FE}}=1.05$ and $\lambda_{\text{RE}}=0.82$).
Both the decoupling FE and RE approaches identified the four loci as significant
($P < 1.2\times10^{-7}$, Bonferroni corrected for $397,450$ tests)
that were reported in the combined analysis of the WTCCC study \cite{Burton07}
(Table \ref{table:wtccc}).
This shows that our approach 
was able to reproduce the previously reported results only using the summary statistics.

Moreover, our decoupling approach has an advantage over the combined analysis \cite{Burton07}. 
Since our approach can utilize both FE and RE, one can have good power to detect
both the homogeneous and the heterogeneous effects \cite{Han11}.
By contrast, the combined analysis can be thought of as similar to FE 
and may not have good power to detect heterogeneous effects.
In the Manhattan plot (Figure \ref{fig:wtcccmht}), 
the notable peaks are the \emph{PTPN22} gene
in the chromosome 1 
and the major histocompatibility complex (MHC) region in the chromosome 6.
Both loci are known to play an important role in autoimmune diseases \cite{Eyre12,Jostins12}.
At both loci, the decoupling RE yields more significant p-values
than the decoupling FE 
(\emph{PTPN22}: $P_{\text{FE}}=5\times 10^{-23}$ and $P_{\text{RE}}=1\times10^{-29}$,
MHC: $P_{\text{FE}}=5\times 10^{-80}$ and $P_{\text{RE}}=8\times 10^{-181}$).
This is because of
the heterogeneous nature of these two loci that 
they are strongly associated to RA and T1D but weakly to CD
(Table \ref{table:wtccc}).
This shows that our decoupling approach can allow one to flexibly apply
different meta-analytic approaches
that are optimized for different situations.

Finally, we examine the robustness of our decoupling approach
by adding the other four diseases as noisy data into the meta-analysis.
Figure \ref{fig:wtcccmht}B shows that when we meta-analyze all seven diseases,
the absolute magnitudes of the significance of
\emph{PTPN22} and MHC loci are reduced, 
but they are still significant for both the decoupling FE and the decoupling RE.
Moreover, the relative significance gain of the decoupling RE compared to the decoupling FE is still largely pronounced
(\emph{PTPN22}: $P_{\text{FE}}=1\times10^{-11}$ and $P_{\text{RE}}=9\times10^{-18}$,
MHC: $P_{\text{FE}}=1\times10^{-28}$ and $P_{\text{RE}}=1\times10^{-89}$).

\section{Materials and Methods}

\subsection{Meta-analytic methods}

We first briefly describe some of the existing meta-analytic methods.

\subsubsection*{Fixed effects model}

The fixed effects model approach (FE) assumes that the magnitude of
the effect size is fixed across the studies \cite{Cochran54,Mantel59}.
The two widely used methods are 
the inverse-variance-weighted effect size method \cite{Fleiss93} and
the weighted sum of z-scores method \cite{Bakker08}.
Since the two methods are approximately equivalent \cite{Han11},
we only describe the inverse-variance-weighted effect size method.
Let $X_1, ..., X_N$ be the effect size estimates in $N$ independent studies
such as the log odds ratios or regression coefficients.
Let $\textrm{SE}(X_i)$ be the standard error of $X_i$ and let $V_i = \textrm{SE}(X_i)^2$.
Let $W_i=V_i^{-1}$ be the inverse variance.
The inverse-variance-weighted summary effect size is
\begin{equation}
  \label{eq:fex}
X_{\text{FE}}=\frac{\sum W_iX_i}{\sum W_i}\;.
\end{equation}
The variance of $X_{\text{FE}}$ is
\begin{equation}
  \label{eq:fev}
 V_{\text{FE}} = \frac{1}{\sum W_i} \;.
\end{equation}
Since the standard error of $X_{\text{FE}}$ is
$\textrm{SE}(X_{\text{FE}})=\sqrt{\sum W_i}^{-1}$,
we can construct a summary z-score
\[Z_{\text{FE}}=\frac{X_{\text{FE}}}{\textrm{SE}(X_{\text{FE}})} 
= \frac{\sum W_iX_i}{\sqrt{\sum W_i}} \]
that follows ${N}(0,1)$ under the null hypothesis
of no associations.
The p-value is calculated
\begin{equation*}
 p_{\text{FE}} = 2\Phi\big(-|Z_{\text{FE}}|\big) 
\end{equation*}
where $\Phi$ is the cumulative density function of the standard normal distribution.

\subsubsection*{Random effects model}

In the random effects (RE) model, it is assumed that the effect size varies among studies, 
a phenomenon called heterogeneity. RE assumes that the effect size follows the probability 
distribution with variance $\tau^2$. There are several approaches to estimate the variance 
$\tau^2$, the most common one being the moment-based estimator of DerSimonian and Laird 
\cite{DerSimonian86}.
 Given the estimate ($\hat{\tau}^2$) of $\tau^2$, the summary effect size estimate is calculated as 
\begin{equation*}
  \label{eq:re_effectsize}
X_{\text{RE}}=\frac{\sum \left(W^{-1}_i + {\hat{\tau}}^2 \right)^{-1}X_i}{\sum \left(W^{-1}_i + {\hat{\tau}}^2 \right)^{-1}}
\end{equation*}
The standard error of $X_{\text{RE}}$ is $\textrm{SE}(X_{\text{RE}}) = \sqrt{\sum \left(W^{-1}_i + {\hat{\tau}}^2 \right)^{-1}}^{-1}$.
In the traditional RE approach, one constructs a z-score statistic 
\begin{equation*}
  \label{eq:re_stat}
Z_{\text{RE}}=\frac{X_{\text{RE}}}{\textrm{SE}(X_{\text{RE}})}
\end{equation*}
and the p-value is computed as $p_{\text{RE}} = 2\Phi\big(-|Z_{\it RE}|\big)$. 
Recently, Han and Eskin found that the traditional RE is conservative and rarely
achieves higher power than the fixed effects model \cite{Han11}.
This is because of the conservatie null hypothesis of the traditional RE
that implicitly assumes heterogeneity under the null.
They proposed a new random effects model that corrects for this problem, and the test statistic is 
\begin{equation}
S_{\text{NewRE}} = \sum \text{log} \left( \frac{V_i}{V_i+\hat{\tau}^2} \right) + \sum \frac{X_i^2}{V_i} - \sum \frac{\left(X_i - \hat{\mu} \right)^2}{V_i+\hat{\tau}^2}
\end{equation}
where $\hat{\mu}$ and $\hat{\tau}^2$ are the maximum likelihood estimations of mean and variance 
of the effect size, respectively, which can be found by an iterative procedure.
The statistic follows a half and half mixture of $\chi^2_{(0)}$ and $\chi^2_{(1)}$ under the null \cite{self87}.
The p-value can be calculated by using the asymptotic distribution
or using a pre-constructed p-value table for a more accurate calculation accounting for 
the small number of observations \cite{Han11}.

\subsubsection*{P-value-based approaches}

There exist meta-analytic approaches combining p-values instead of effect sizes, the most 
traditional one being the Fisher's method \cite{Fisher67}.
The Fisher's method combines the p-values $p_1, ..., p_N$ by constructing a statistic
\[ S_{\text{Fisher}} = -2 \sum_{i=1}^N \log(p_i) \]
which follows $\chi^2_{(2N)}$ under the null.
The p-value-based approaches have advantages that it can be used even when one does not 
have information about the direction of effects. 
A recently proposed p-value-based approach 
by Zaykin and Kozbur
can take into account overlapping subjects \cite{Zaykin10}.

\subsubsection*{Subset approaches}

Subset approaches have similarities to the fixed effects model, 
but the difference is that one assumes that the effects can exist in only a subset of the studies.
Han and Eskin computes a statistic called ``m-value'', a posterior probability that the effect is present
in a study \cite{Han12}.
The m-values are incorporated 
in the weighted z-score FE approach to
upweight studies with high m-values.
The p-value is assessed using the importance sampling.
Bhattacharjee et al. \cite{Bhattacharjee12} proposed an approach that computes FE statistics using all possible subsets of 
studies in a meta-analysis and uses the maximum statistic. 
The method expedites the enumeration of all possible $2^N-1$ possible subsets of $N$ studies
by using a novel statistical technique. 
This method can account for overlapping subjects.

\subsubsection*{Bayesian approaches}

Morris proposed a Bayesian approach optimized for the trans-ethnic meta-analysis \cite{Morris11}.
This method utilizes the Markov Chain Monte Carlo (MCMC) procedure to navigate through possible
disease models.
In his MCMC, 
closely related populations 
have higher chance to have similar effect sizes to increase the statistical power to detect
heterogeneity caused by the population spectrum.
Wen \cite{Wen11} proposed a new method taking into account heterogeneity in the data
using the hierarchical model in the Bayesian framework.
Wen \cite{Wen12} recently extended this method to a multi-way table modeling,
which can account for the correlation structure between studies
or the overlapping subjects.

\subsection{Lin and Sullivan approach}

Lin and Sullivan proposed a systematic approach for dealing with overlapping subjects
for the fixed effects model \cite{Lin09}.
The first step of their approach is to analytically calculate the correlation 
of the statistics $X_1, ..., X_N$ caused by the overlapping subjects.
Let $\mathbf{x}$ be the vector of effect size estimates $\mathbf{x}=(X_1, ..., X_N)$
and let \[ \mathbf{C} = \big[ r_{ij} \big]_{N \times N}\]
be the correlation matrix of $X$ 
where $r_{ij}$ denotes the the correlation between $X_i$ and $X_j$.
$r_{ij}$ is analytically approximated with the formula
\begin{equation}
\label{eq:r}
r_{ij} \approx \left(n_{ij0}\sqrt{\frac{n_{i1}n_{j1}}{n_{i0}n_{j0}}}
+n_{ij1}\sqrt{\frac{n_{i0}n_{j0}}{n_{i1}n_{j1}}}\right) \Big/ \sqrt{n_in_j}
\end{equation}
$n_{i1}$, $n_{i0}$, and $n_i$ (or $n_{j1}$, $n_{j0}$, and $n_j$) are
the number of cases, the number of controls, and the total number of subjects in the $i$th (or $j$th) study, 
respectively.
$n_{ij1}$ and $n_{ij0}$ are the numbers of cases and controls that overlap between the $i$th and $j$th studies. 
See Bhattacharjee et al. \cite{Bhattacharjee12} for an extended formula for the situation
that some cases in one study are controls in other studies.
Given the correlation matrix $\mathbf{C}$,
it is straightforward to calculate
the covariance matrix $\mathbf{\Omega}$.

The second step of the Lin and Sullivan approach is to optimally take into account
$\mathbf{\Omega}$ in the testing.
The optimal fixed effects model meta-analysis statistic is
\begin{equation}
\label{eq:linx}
 X_{\text{Lin}} = \frac{\mathbf{e}^{T}\mathbf{\Omega}^{-1}\mathbf{x}}{\mathbf{e}^{T}\mathbf{\Omega}^{-1}\mathbf{e}} 
\end{equation}
where $\mathbf{e}$ is the vector of ones ($\mathbf{e} = (1, ..., 1)$).
The formal proof for the optimality of this statistic is shown in 
\cite{Wei89} and \cite{Wei85}.
We also present a simple reasoning for deriving this statistic
in Supplementary Materials.
The variance of the statistic is given
\begin{equation}
\label{eq:linv}
  \var{X_{\text{Lin}}} = \frac{1}{\mathbf{e}^{T}\mathbf{\Omega}^{-1}\mathbf{e}} 
\end{equation}
The z-score $X_{\text{Lin}}/\sqrt{\var{X_{\text{Lin}}}}$ is calculated to obtain the p-value.
Note that in a special case that $X_1, ..., X_N$ are independent
($\mathbf{\Omega}$ is a diagonal matrix), 
$X_{\text{Lin}}$ equals to $X_{\text{FE}}$ 
and $\var{X_{\text{Lin}}}$ equals to
$\var{X_{\text{FE}}}$.

\subsection{Decoupling approach}

We extend the Lin and Sullivan approach to a general framework that can be applied
to a wide range of meta-analytic methods.
Our approach ``uncorrelates'' or ``decouples'' correlated studies into independent studies
whose standard errors are updated to account for the decoupling.
Suppose that we are given the effect sizes $\mathbf{x}$, the standard errors of them $\mathbf{s}$, and
the correlation matrix $\mathbf{C}$ computed by the formula (\ref{eq:r}).
The decoupling procedure is the following.
\begin{enumerate}
\item Keep the original $\mathbf{x}$. \\ 
$\mathbf{x}_{\text{Decoupled}} \leftarrow \mathbf{x}$ 
\item Compute the covariance matrix of $\mathbf{x}$.\\
 $\mathbf{\Omega} \leftarrow Diag(\mathbf{s})\cdot \mathbf{C}\cdot Diag(\mathbf{s})$
\item Compute the decoupled covariance matrix. \\
$\mathbf{\Omega}_{\text{Decoupled}} \leftarrow Diag(\mathbf{e}^T\mathbf{\Omega}^{-1})^{-1}$
\item Update the standard errors. \\
$\mathbf{s}_{\text{Decoupled}}[i] \leftarrow \sqrt{\mathbf{\Omega}_{\text{Decoupled}}[i,i]} $
for each $i=1, ..., N$
\item Use $\mathbf{x}_{\text{Decoupled}}$ and 
$\mathbf{s}_{\text{Decoupled}}$ in the downstream meta-analysis.
\end{enumerate}
$Diag(\mathbf{s})$ denotes a diagonal matrix whose diagonals are $\mathbf{s}$.
The brakets [\;] denote the index of an element of a vector or a matrix.
In Supplementary Materials, we present a simple R code performing this procedure.

We give a simple working example of this procedure.
Suppose that we have two effect sizes $X_1$ and $X_2$.
For simplicity, let their variances be 1.0.
Under the fixed effects model, the best summary estimate of effect size
will be $(X_1+X_2)/2$ and its variance will be $1.0/2=0.5$,
which will be the correct variance if the two studies are independent ($r_{12}=0$).
Now consider the case that $X_1$ and $X_2$ are highly correlated ($r_{12}=0.99$).
Intuitively, since they are highly correlated, the information they contain is not much
better than the information a single study contains.
The decoupling formula gives us the new variance of each study increased to $1.99$.
When we plug the new variances into the fixed effects model,
the variance of the summary effect size will be $1.99/2 = 0.995 \approx 1.0$
showing that as expected, the uncertainty in the final estimate is approximately
the same as the uncertainty that we would obtain with a single study.

\begin{observation}
Using the decoupled studies in the fixed effects model
is equivalent to the Lin and Sullivan approach.
\end{observation}
\begin{proof}
Given a covariance matrix $\mathbf{\Omega}$, our decoupling approach will calculate
the updated standard errors $\mathbf{s}_{\text{Decoupled}}$ by calculating $\mathbf{\Omega}_{\text{Decoupled}}$.
Since $\mathbf{\Omega}_{\text{Decoupled}}$ is a diagonal matrix, the following relationship holds
\[\mathbf{\Omega}_{\text{Decoupled}} = Diag(\mathbf{s}_{\text{Decoupled}}) \cdot Diag(\mathbf{s}_{\text{Decoupled}})\]
On the other hand,
given an effect size vector $\mathbf{x}$ and standard errors $\mathbf{s}$,
the standard fixed effects model formulae in equation (\ref{eq:fex}) and (\ref{eq:fev})
can be written 
\[ X_{\text{FE}} = \frac{\mathbf{e}^{T}\mathbf{V}^{-1}\mathbf{x}}{\mathbf{e}^{T}\mathbf{V}^{-1}\mathbf{e}} \]
and
\[  V_{\text{FE}} = \frac{1}{\mathbf{e}^{T}\mathbf{V}^{-1}\mathbf{e}} \]
where $\mathbf{V} = Diag(\mathbf{s}) \cdot Diag(\mathbf{s})$.
If we plug $\mathbf{s}_{\text{Decoupled}}$ into this formula, 
\begin{align*}
  X_{\text{FE}} 
&= \frac{\mathbf{e}^{T}\mathbf{V}^{-1}\mathbf{x}}{\mathbf{e}^{T}\mathbf{V}^{-1}\mathbf{e}} \\
&= \frac{\mathbf{e}^{T}
(Diag(\mathbf{s}_{\text{Decoupled}}) \cdot Diag(\mathbf{s}_{\text{Decoupled}}))^{-1}
\mathbf{x}}{\mathbf{e}^{T}
(Diag(\mathbf{s}_{\text{Decoupled}}) \cdot Diag(\mathbf{s}_{\text{Decoupled}}))^{-1}
\mathbf{e}} \\
&= \frac{\mathbf{e}^{T}
\mathbf{\Omega}_{\text{Decoupled}}^{-1}
\mathbf{x}}{\mathbf{e}^{T}
\mathbf{\Omega}_{\text{Decoupled}}^{-1}
\mathbf{e}} \\
&= \frac{\mathbf{e}^{T}
Diag(\mathbf{e}^T\mathbf{\Omega}^{-1})
\mathbf{x}}{\mathbf{e}^{T}
Diag(\mathbf{e}^T\mathbf{\Omega}^{-1})
\mathbf{e}} \\
&= \frac{\mathbf{e}^{T}
\mathbf{\Omega}^{-1}
\mathbf{x}}{\mathbf{e}^{T}
\mathbf{\Omega}^{-1}
 \mathbf{e}} \\
  &= X_{\text{Lin}}
\end{align*}
where $X_{\text{Lin}}$ is in equation (\ref{eq:linx}).
Similarly, we can show that $V_{\text{FE}}$ equals to $V_{\text{Lin}}$ in equation (\ref{eq:linv}).
Therefore, applying our decoupling approach to the fixed effects model is equivalent
to the Lin an Sullivan approach \cite{Lin09}.

\end{proof}

\section{Discussions}

We proposed a general framework for dealing with overlapping subjects in a meta-analysis.
The core idea is to ``decouple'' the correlated studies into independent studies
and use them in the downstream meta-analysis.
Our approach can flexibly allow many meta-analytic methods,
such as the random effects model, to account for overlapping subjects.
The simulations and the applications to the WTCCC data support
the utilities of our approach.

Since our approach involves only the data-side change rather than the method-side change,
one advantage is that 
one can continue to use preferred software in the existing analysis pipeline
which may not support overlapping subjects.
This can be important in practice, 
since it can be inconvenient to change a pipeline.
For example, one may have been using METAL \cite{Metal:web} or MANTEL \cite{Bakker08}
for automatically detecting strand inconsistencies between studies
and for automatically applying the genomic control \cite{Devlin95}.
Given new data with overlapping subjects,
one does not need to switch to different software supporting overlapping subjects
but can simply update the standard errors using our approach
and continue to use the existing pipeline.

In this paper, we primarily focused on dealing with overlapping subjects,
but our decoupling approach can be applied to any contexts of meta-analysis
where the inputs are correlated.
For example, in an eQTL study, multiple tissues can be analyzed together in a meta-analysis framework
\cite{Sul13}.
Since tissues of the same individual are correlated, 
this results in a meta-analytic problem where the inputs are correlated.
In such cases, our decoupling approach with FE and RE methods can be applied
to detect both tissue-specific and shared eQTLs \cite{Sul13}.

The limitation of our approach is that 
the optimality is guaranteed only under the fixed effects model.
For example, our simulations show that our approach 
has some conservative tendencies under the random effects model,
indicating that our approach may not be optimal, although we showed that
it works well in the simulations and the WTCCC data.
Optimal solutions to account for correlated inputs for each different 
meta-analytic method will be an interesting 
topic for further research.

We note that one should be careful in interpreting data based on the decoupled studies.
For example, the heterogeneity testing is highly conservative
when using decoupled studies and may not well detect true heterogeneity.
Unfortunately, there is no good alternative to the Cochran's Q test 
and the $I^2$ estimate that can take into account correlated inputs yet.
Developing such methods will also be an interesting future research area.

\bibliography{2013_04_23_papers2}

\clearpage
\section{Tables and Figures}

\begin{table}[h]
  \centering
  \begin{tabular}{llc}
Category&    Method & Supports overlapping subjects \\
\hline
Fixed effects model & Traditional FE \cite{Cochran54,Mantel59}& No \\
& Lin and Sullivan \cite{Lin09} & Yes \\
Random effects model & Traditional RE \cite{DerSimonian86}& No \\
& Han and Eskin RE \cite{Han11}& No \\
P-value-based approaches  & Fisher's approach \cite{Fisher67} & No\\
& Zaykin and Kozbur \cite{Zaykin10} & Yes \\
Subset approaches & Bhattacharjee et al. \cite{Bhattacharjee12} & Yes \\
& Binary Effects Model \cite{Han12} & No \\
Bayesian approaches & Hierarchical Bayesian approach \cite{Wen11} & No \\
&Trans-ethnic approach  \cite{Morris11} & No \\
&Multi-way table approach \cite{Wen12} & Yes \\
\end{tabular}
\caption{
  \label{table:methods}
Different meta-analytic methods and their support for overlapping subjects.
}
\end{table}

\begin{sidewaystable}[p]
  \centering
  \begin{tabular}{llllllllll}
&&&\multicolumn{3}{c}{Single disease P-values} && \multicolumn{2}{c}{Meta-analysis P-values}&\\
\cline{4-6} \cline{8-9}
Chr. & RSID & Position & CD & RA & T1D && Decouple-FE &  Decouple-RE & Genes\\
\hline
1&rs6679677&114015850&0.0372&$4.46\times10^{-16}$&$2.20\times10^{-17}$&&$4.48\times10^{-23}$&$1.00\times10^{-29}$&{\it PTPN22} \\
6&rs9273363&32734250&0.225&$5.18\times10^{-07}$&$8.14\times10^{-180 }$&&$5.02\times10^{-80}$&$7.72\times10^{-181}$& HLA genes \\
12&rs17696736&110949538&0.0189&0.0281&$1.98\times10^{-09 }$&&$7.72\times10^{-10}$&$3.40\times10^{-10}$& {\it SH2B3}\\
18&rs2542151&12769947&$3.66\times10^{-05}$&0.0841&0.000297&&$3.74\times10^{-08}$&$5.16\times10^{-08}$& {\it PTPN2}\\
\end{tabular}
\caption{
  \label{table:wtccc}
  Association results of a meta-analysis combining 
  CD, RA, and T1D of the WTCCC data.
  The four loci reported in the WTCCC combined 
  analysis \cite{Burton07} are presented.
  The curated list of candidate genes for each loci are obtained from the ImmunoBase web site (\texttt{www.ImmunoBase.org}).
}
\end{sidewaystable}

\begin{figure}[p]
  \centering
  \includegraphics[width=.75\textwidth]{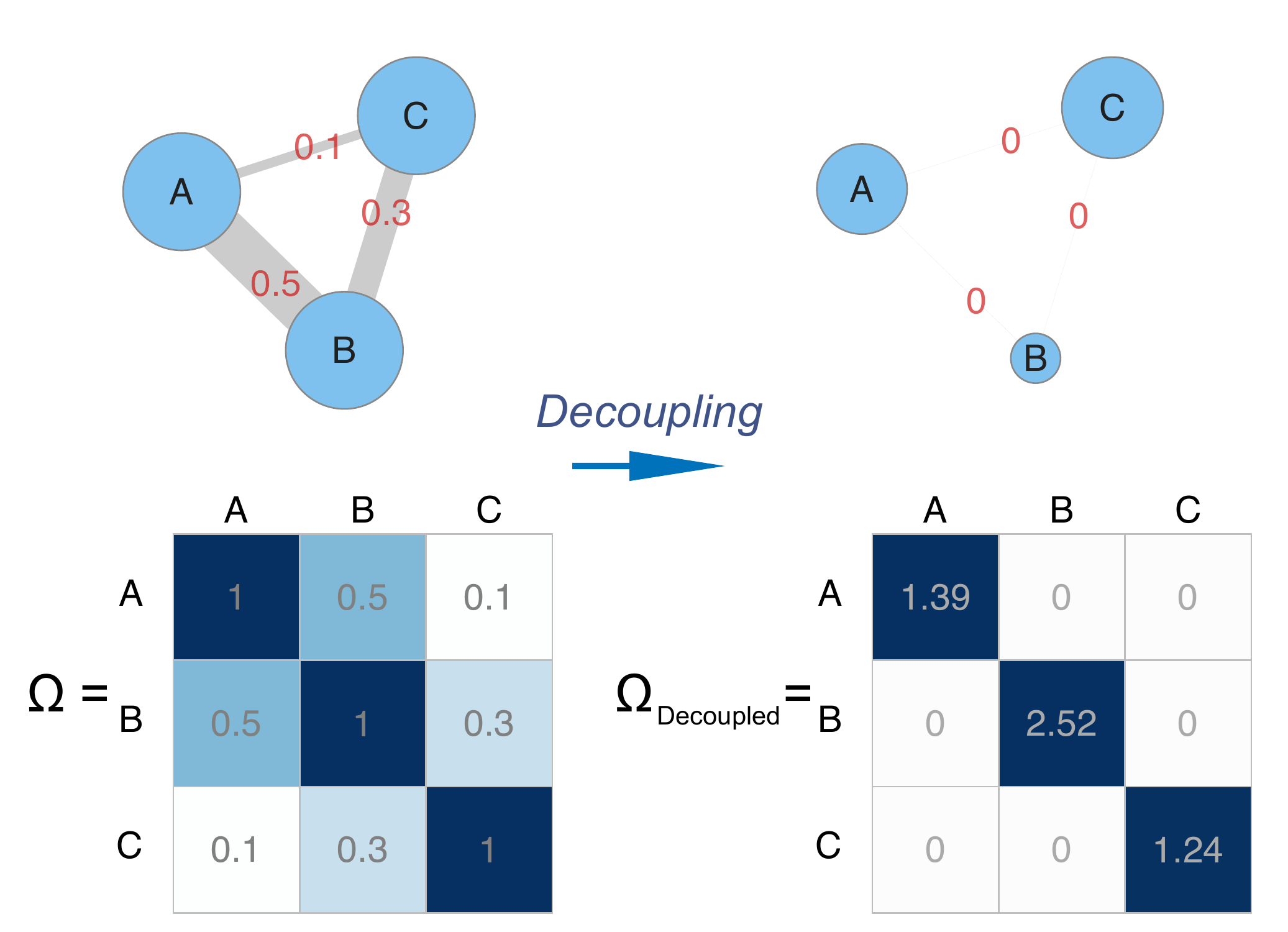}
  \caption{
  \label{fig:schematic}
A simple example of our decoupling approach.
$\mathbf{\Omega}$ and $\mathbf{\Omega}_{\text{Decoupled}}$ are the covariance matrices of the statistics of three studies
A, B, and C before and after decoupling respectively.
The thickness of the edges denotes the amount of correlation between the studies.
After decoupling, the size of the nodes reflects the information 
that the studies contain in terms of
the inverse variance.
}
\end{figure}

\begin{figure}[p]
  \centering
  \includegraphics[width=.95\textwidth]{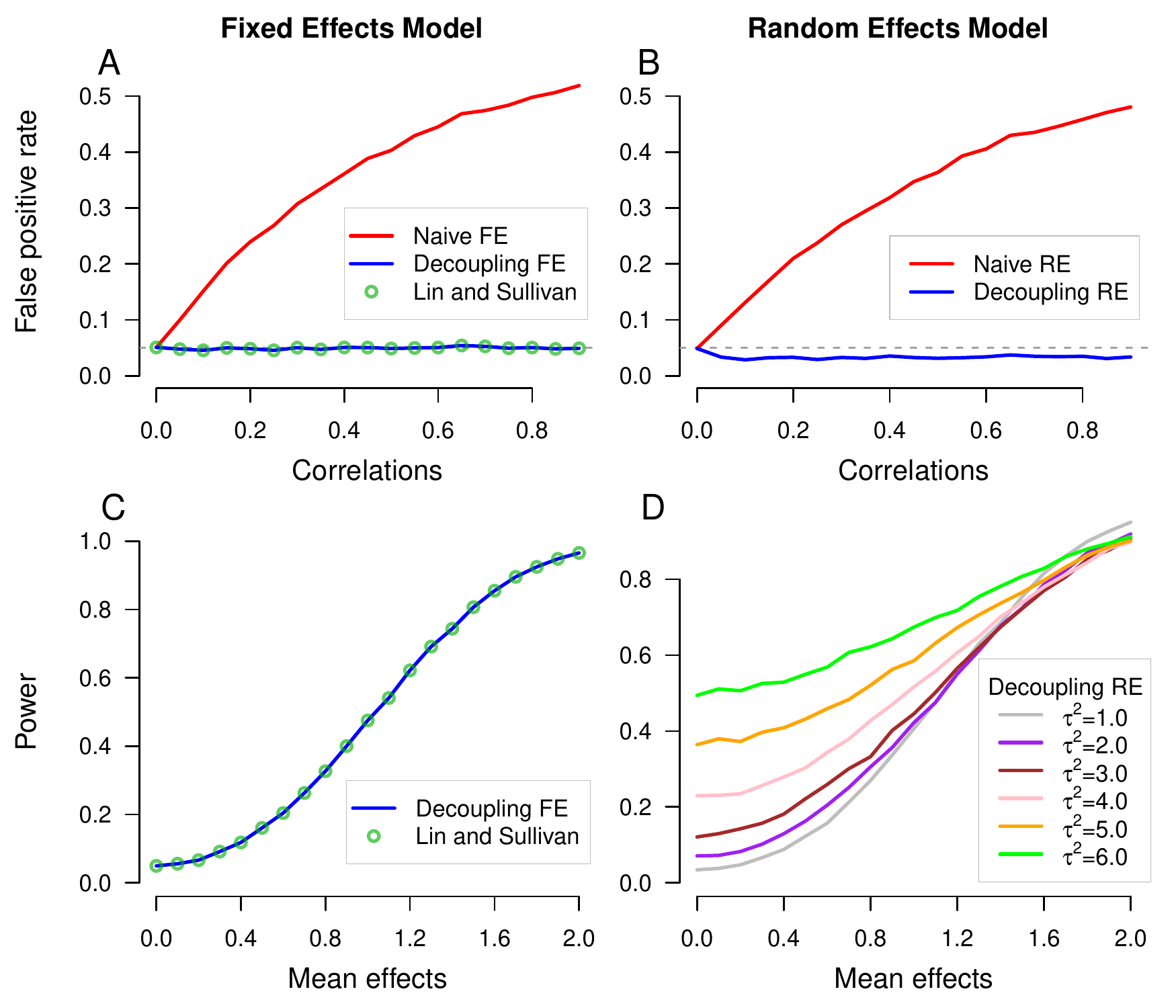}
  \caption{
  \label{fig:power}
The false positive rate and power of different methods under the fixed and the random effects model.
(A, B) We simulate studies assuming the null model of no associations.
(C) We simulate studies assuming the alternative model based on the fixed effects model 
where the effect sizes are fixed across studies. 
We vary the magnitude of the fixed effect size.
(D) We simulate studies assuming the alternative model based on the random effects model 
where the effect sizes vary across studies with an additional variance $\tau^2$.
We vary both the magnitude of the mean effect size and the heterogeneity $\tau^2$.
}
\end{figure}

\begin{figure}[p]
  \centering
  \includegraphics[width=.65\textwidth]{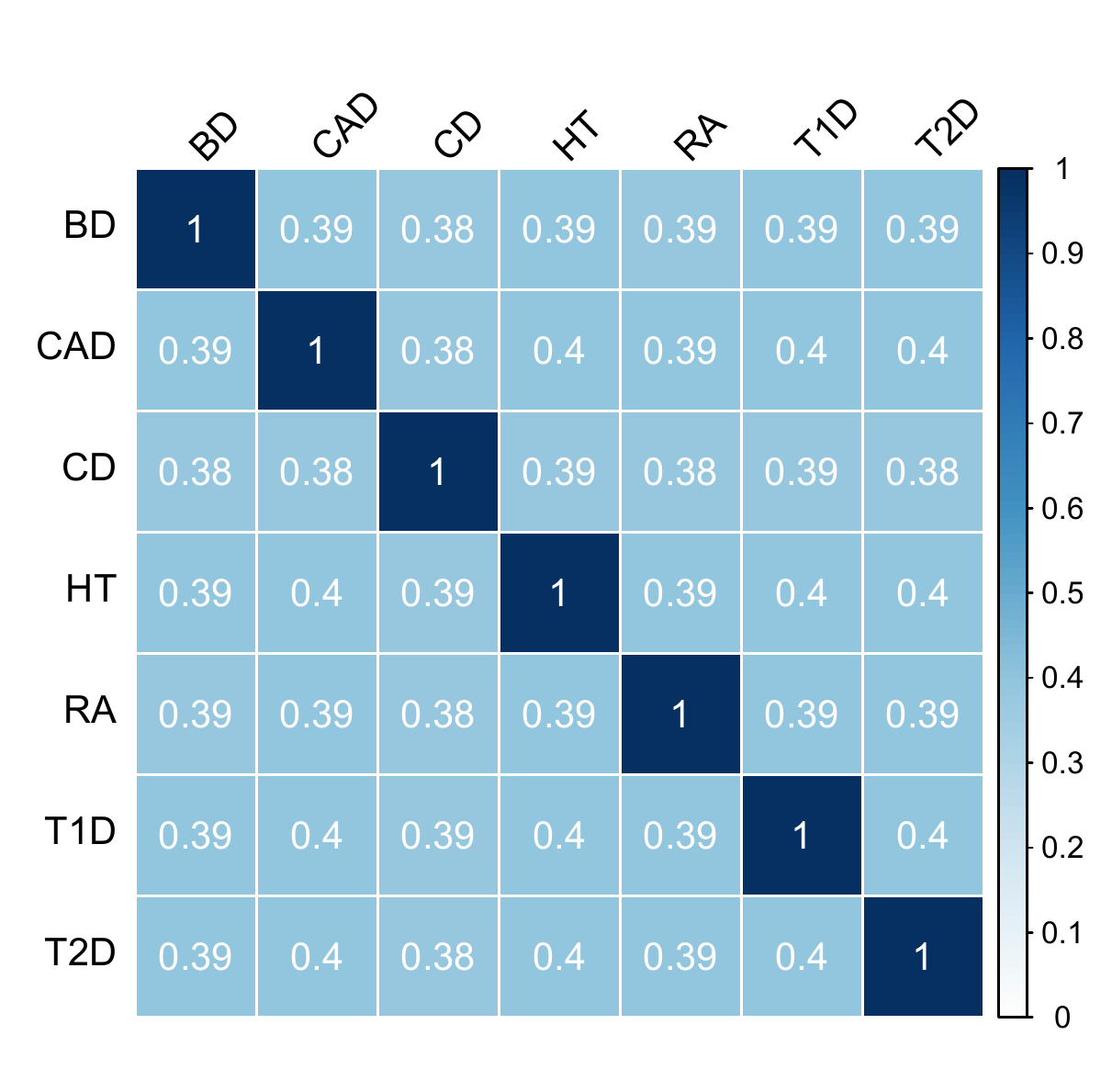}
  \caption{
  \label{fig:cor}
The correlation structure of the statistics of seven diseases in the WTCCC data. 
All seven diseases have approximately 2,000 cases and share 2,938 controls.
}
\end{figure}

\begin{figure}[p]
  \centering
  \includegraphics[width=1.0\textwidth]{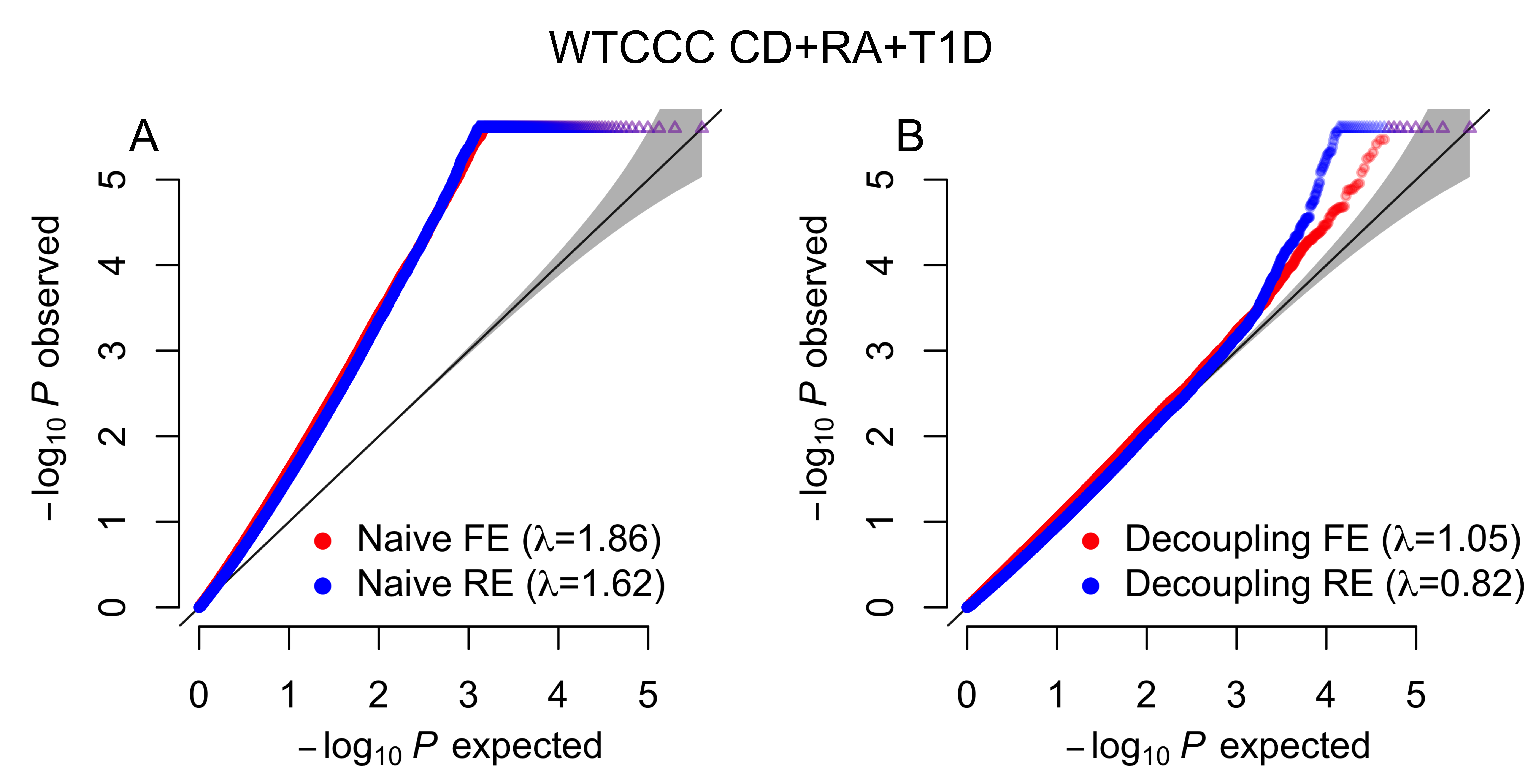}
  \caption{
    \label{fig:wtcccqq}
    The qq-plots of a meta-analysis combining 
    CD, RA, and T1D of the WTCCC data.
    The MHC region is excluded.
}
\end{figure}

\begin{figure}[p]
  \centering
  \includegraphics[width=1.0\textwidth]{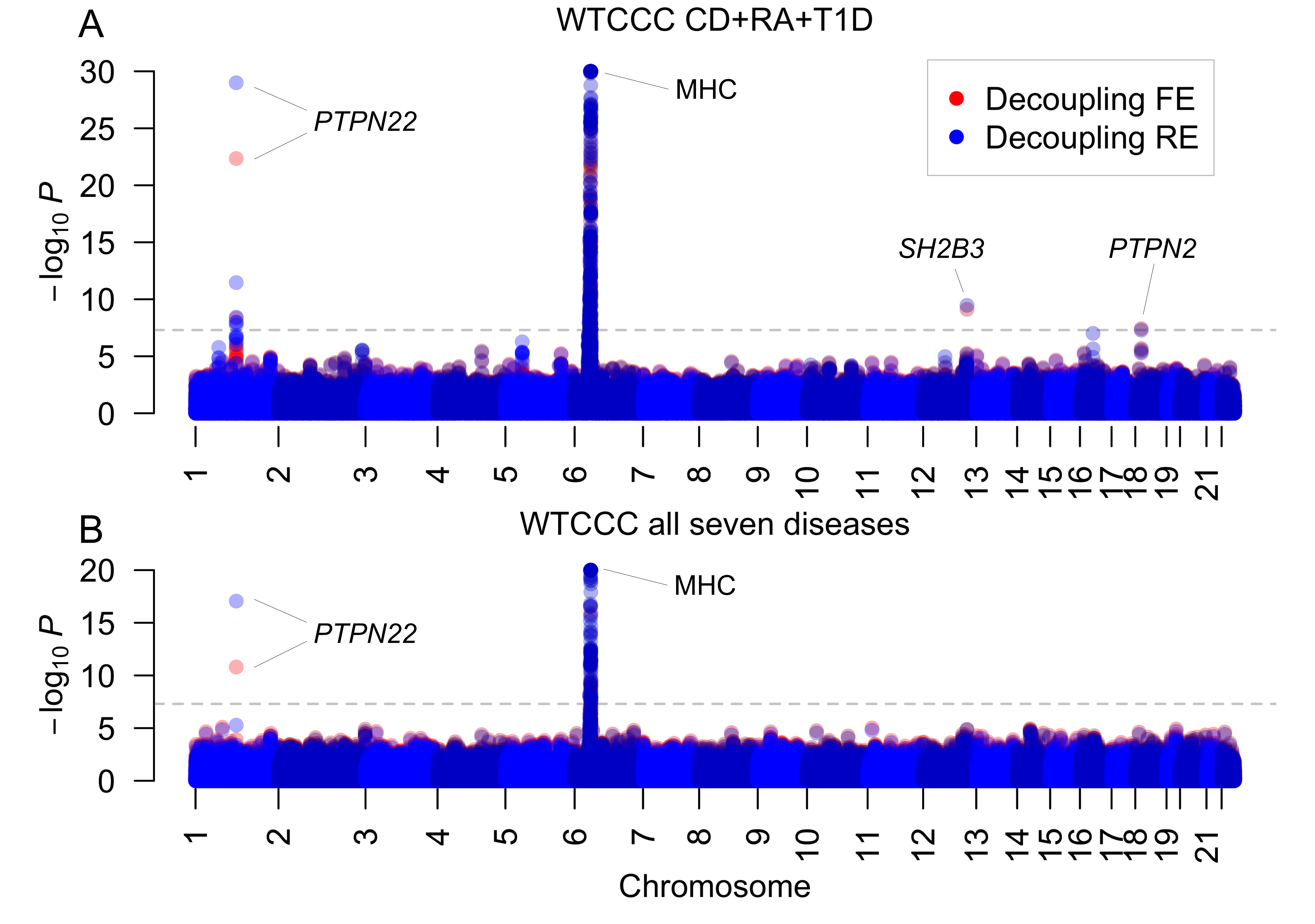}
  \caption{
    \label{fig:wtcccmht}
    The Manhattan plots of meta-analyses combining a set of diseases
    of the WTCCC data.
    (A) We combine CD, RA, and T1D. (B) We combine all seven diseases.
}
\end{figure}

\clearpage
\section{Supplementary Materials}

\subsection{Optimality of Lin and Sullivan approach}

We present a simple reasoning to derive 
and show the optimality of the Lin and Sullivan statistic
in equation (\ref{eq:linx}).
Note that the formal proof of the optimality is shown in \cite{Wei85} and \cite{Wei89},
and there can be many other proofs.
The equation (\ref{eq:linx}) is
\begin{equation*}
 X_{\text{Lin}} = \frac{\mathbf{e}^{T}\mathbf{\Omega}^{-1}\mathbf{x}}{\mathbf{e}^{T}\mathbf{\Omega}^{-1}\mathbf{e}} 
\end{equation*}
where the notations are described in Methods.
One simple way to derive this statistic is to translate the problem of finding
a summary statistic to a linear regression framework.
Consider that $\mathbf{x}$ is the dependent variables whose variance is $\mathbf{\Omega}$.
Then finding the best summary statistic is equivalent to finding the best mean or the intercept $\beta$.
Thus, we have a regression model including only the intercept term 
\[ \mathbf{x} = \beta \mathbf{e} + \epsilon \]
where $\var{\epsilon} = \mathbf{\Omega}$. 
By the standard generalized least square formula, 
the optimal estimate of $\beta$ will be 
\[ \hat{\beta} = (\mathbf{e}^T \mathbf{\Omega}^{-1} \mathbf{e})^{-1} \mathbf{e}^T \mathbf{\Omega}^{-1} \mathbf{x} \]
Since $\mathbf{e}^T \mathbf{\Omega}^{-1} \mathbf{e}$ is a scalar, this form is equivalent to 
$X_{\text{Lin}}$.

\subsection{R code performing decoupling approach}
\begin{verbatim}
## Decoupling approach.
## Input:
##   s is the standard errors (possibly with NA)
##   C is the correlation matrix
##   (or one can specify C.inv (inverse of C) for speed-up)
## Output:
##   updated standard errors after decoupling
decoupling <- function(s, C, C.inv=NULL) {
  i <- !is.na(s)
  if (is.null(C.inv)) {
    Omega.inv <- solve(diag(s[i]) %*% C[i,i] %*% diag(s[i]))
  } else {
    Omega.inv <- diag(1/s[i]) %*% C.inv[i,i] %*% diag(1/s[i])
  }
  s.new <- sqrt(1/rowSums(Omega.inv))
  s[i] <- s.new
  s
}
\end{verbatim}
\end{document}